\documentclass[11pt,letterpaper]{article}
\usepackage{fullpage}
\usepackage[top=2cm, bottom=4.5cm, left=2.5cm, right=2.5cm]{geometry}
\usepackage{amsmath,amsthm,amsfonts,amssymb,amscd}
\usepackage{lastpage}
\usepackage{enumerate}
\usepackage{fancyhdr}
\usepackage{mathrsfs}
\usepackage{xcolor}
\usepackage{graphicx}
\usepackage{listings}
\usepackage[colorlinks=true, urlcolor=blue, linkcolor=blue, citecolor=magenta]{hyperref}
\usepackage{cite}
\usepackage{multirow}
\usepackage{comment}
\usepackage{mathpazo}

\hypersetup{%
  colorlinks=true,
  linkcolor=red,
  linkbordercolor={0 0 1}
}

\lstdefinestyle{Matlab}{
    language        = matlab,
    frame           = lines, 
    basicstyle      = \footnotesize,
    keywordstyle    = \color{blue},
    stringstyle     = \color{green},
    commentstyle    = \color{red}\ttfamily
}
\newtheorem{theorem}{Theorem}
\newtheorem{definition}[theorem]{Definition}
\newtheorem{claim}[theorem]{Claim}

\newtheorem{remark}[theorem]{Remark}
\newtheorem{lemma}[theorem]{Lemma}
\newtheorem{hypothesis}[theorem]{Hypothesis}
\newtheorem{corollary}[theorem]{Corollary}

\newenvironment{restated_thm}[1]
{\renewcommand{\thetheorem}{\ref{#1}}%
 \addtocounter{theorem}{-1}%
 \begin{theorem}}
{\end{theorem}}
\newenvironment{restated_cor}[1]
{\renewcommand{\thetheorem}{\ref{#1}}%
 \addtocounter{theorem}{-1}%
 \begin{corollary}}
{\end{corollary}}

\begin{document}

\title{On Lower Bounds of Approximating Parameterized $k$-Clique}
\author{
Bingkai Lin \thanks{Nanjing University. Email: \texttt{lin@nju.edu.cn}}
\and
Xuandi Ren\thanks{Peking University. Email: \texttt{renxuandi@pku.edu.cn}}
\and
Yican Sun\thanks{Peking University. Email: \texttt{sycpku@pku.edu.cn}}
\and 
Xiuhan Wang\thanks{Tsinghua University. Email: \texttt{wangxh19@mails.tsinghua.edu.cn}}
}

\maketitle

\begin{abstract}
    Given a simple graph $G$ and an integer $k$, the goal of {\sc $k$-Clique} problem is to decide if $G$ contains a complete subgraph of size $k$. We say an algorithm approximates {\sc $k$-Clique} within a factor $g(k)$ if it can find a clique of size at least $k / g(k)$ when $G$ is guaranteed to have a $k$-clique. Recently, it was shown that approximating {\sc $k$-Clique} within a constant factor is {\sf W[1]}-hard \cite{Lin21}.
    
    We study the approximation of {\sc $k$-Clique} under the \emph{Exponential Time Hypothesis} ({\sf ETH}). The reduction of \cite{Lin21} already implies an $n^{\Omega(\sqrt[6]{\log k})}$-time lower bound under {\sf ETH}. We improve this lower bound to $n^{\Omega(\log k)}$. Using the gap-amplification technique by expander graphs, we also prove that there is no $k^{o(1)}$ factor FPT-approximation algorithm for {\sc $k$-Clique} under {\sf ETH}.
    
    We also suggest a new way to prove the \emph{Parameterized Inapproximability Hypothesis} ({\sf PIH}) under {\sf ETH}. We show that if there is no $n^{O(\frac{k}{\log k})}$ algorithm to approximate {\sc $k$-Clique} within a constant factor, then {\sf PIH} is true.
\end{abstract}

\section{Introduction}
\label{sec:intro}

In this paper, we study the {\sc $k$-Clique} problem: given a simple graph $G$ and an integer $k$, decide whether $G$ contains a complete subgraph of size  $k$. As shown in \cite{kar72}, {\sc $k$-Clique} is one of the most classical NP-complete problems. Its inapproximability in the classical complexity regime has also been studied extensively \cite{feige1996interactive,bellare1993efficient,bellare1994improved,goldwasser1998introduction,feige2000two,hastad1996clique,Zuc07}. Along a long line of research, it was proved that even approximating {\sc Clique} into a ratio of $n^{1-\varepsilon}$ is NP-hard.

In recent years, the hardness of approximating {\sc $k$-Clique} has received increased attention in the parameterized complexity regime. When guaranteed that the maximum clique is of size $k$, people wonder if there is an algorithm which runs in $f(k)n^{O(1)}$ time, and can find a clique of size at least $k/g(k)$, for some computable functions $f$ and $g$. Such an algorithm is called a \textit{$g(k)$-FPT-approximation} for the {\sc $k$-Clique} problem.

Previously, \cite{CCK+17} ruled out all $g(k)$-FPT-approximation algorithms of {\sc $k$-Clique} for any $g(k)=o(k)$ under the \textit{Gap Exponential Time Hypothesis} (\textsf{Gap-ETH}) \footnote{{\sf Gap-ETH} states that no subexponential time algorithm can distinguish whether a {\sc 3SAT} formula is satisfiable or every assignment satisfies at most $1 - \varepsilon$ fraction of clauses for some $\varepsilon > 0$.}. They even showed that assuming \textsf{Gap-ETH}, it is impossible to find a clique of size $\varepsilon (k)$ in $f(k)n^{o(\varepsilon(k))}$ time. However, as \textsf{Gap-ETH} is such a strong hypothesis that it already gives a gap in hardness of approximation, it is still of great interest to prove the same lower bound under an assumption without an inherent gap. People may further wonder:

\begin{center}
    \textit{Assuming {\sf ETH}, does finding a clique of size $\varepsilon(k)$ in {\sc $k$-Clique} require $f(k)n^{\Omega(\varepsilon(k))}$ time?}
\end{center}

 In a recent work \cite{Lin21}, Lin showed that {\sc $k$-Clique} does not admit constant factor FPT-approximation algorithms unless $\sf W[1]=FPT$. This was the first successful attempt to bypass \textsf{Gap-ETH} to prove the hardness of approximating {\sc $k$-Clique}. Unfortunately, \cite{Lin21} reduces a {\sc $k$-Clique} instance to a constant gap {\sc $k'$-Clique}\footnote{Given $k\in\mathbb{N}$, $\varepsilon\in(0,1)$ and  a simple graph $G$,  the constant gap {\sc $k$-Clique} problem is to decide whether $G$ contains a $K_k$ subgraph or $G$ contains no $K_{\varepsilon k}$ subgraph.} instance with $k'=2^{k^6}$. As there is no $f(k)n^{o(k)}$ time algorithm for {\sc $k$-Clique} assuming \textsf{ETH}, \cite{Lin21} actually ruled out $f(k) n^{o(\sqrt[6]{\log k})}$ time constant approximation algorithms for  {\sc $k$-Clique} under {\sf ETH}. Comparing to \cite{CCK+17}, such a lower bound is still far beyond satisfaction, and it remains open to avoid the huge parameter blow-up in the gap-producing reduction to obtain a better lower bound.

The main result of this paper is
\begin{theorem}
\label{thm:main}
    Assuming {\sf ETH}, for any constant $c>1$ and any computable function $f$, no algorithm can find a clique of size $k/c$ in the {\sc $k$-Clique} problem in $f(k)n^{o(\log k)}$ time.
\end{theorem}

As an application, we combine our main result with the classical gap-amplification technique to rule out any $k^{o(1)}$-ratio FPT-approximation algorithms for {\sc $k$-Clique} under {\sf ETH}. Let us not fail to mention that \cite{KS21} recently proved  similar lower bound based on a weaker hypothesis $\sf W[1]\neq FPT$. Our result is formally stated as follows.
\begin{corollary}
\label{cor:main}
	Assuming {\sf ETH}, for any $g(k) = k^{o(1)}$, the {\sc $k$-Clique} problem has no $g(k)$-FPT-approximation algorithm.
\end{corollary}

We also study the relationship between the constant gap {\sc $k$-Clique} problem and the \textit{parameterized inapproximablity hypothesis} ({\sf PIH}) \cite{LRSZ20}, a central conjecture in parameterized complexity. 
\ Roughly speaking, \textsf{PIH} states that it is impossible to approximate a 2-CSP instance over $k$ variables with alphabet $[n]$ to a constant factor in FPT time.
It is known in \cite{feige1996interactive} that if \textsf{PIH} is true, then there is no FPT algorithm for constant gap {\sc $k$-Clique}. However, the reverse direction is not known yet. Furthermore, although \textsf{PIH} can be deduced from  \textsf{Gap-ETH} via standard reductions in \cite{CHKX04,ChenHKX06}, proving \textsf{PIH} under gap-free hypotheses (e.g. \textsf{ETH}, $\sf W[1]\neq FPT$) is still quite open and is believed to require a PCP-like theorem in parameterized complexity. 
\ We show that an almost-tight running time lower bound of constant gap {\sc $k$-Clique} could imply {\sf PIH}. Our theorem suggests a new way to prove {\sf PIH} under {\sf ETH}, namely, by using constant gap {\sc $k$-Clique} as an intermediate problem. It is formally stated as follows.

\begin{theorem}\label{thm:gapcliqe2pih}
If there is no $f(k) n^{O\left(\frac{k}{\log k}\right)}$ time algorithm for constant gap {\sc $k$-Clique}, then {\sf PIH} is true.
\end{theorem}

\subsection{Our Techniques}

\noindent\textbf{From {\sc 3SAT} to gap {\sc $k$-Clique}.} 
Recall that the reduction in \cite{Lin21} consists of two steps. First, it reduces {\sc $k$-Clique}  to  {\sc $k^2$-VectorSum}, while introducing a quadratic blow-up of the parameter. Next, it transforms {\sc $k$-VectorSum} to {\sc CSP} on $k'=2^{O(k^3)}$ variables $\{x_{\vec{a}_1,\ldots,\vec{a}_k} : \vec{a}_1,\ldots,\vec{a}_k\in\mathbb{F}^k\}$, and then to constant gap {\sc $k'$-Clique}. The two steps together cause the parameter to grow from $k$ to $2^{O(k^6)}$.

To give a tighter lower bound of constant gap {\sc $k$-Clique} under {\sf ETH}, we deal with the above two steps separately. First, we show a reduction directly from {\sc 3SAT} to {\sc $k$-VectorSum}, resulting in a tighter lower bound of {\sc $k$-VectorSum} under {\sf ETH}. Then, we give a more succinct reduction from {\sc $k$-VectorSum} to {\sc CSP} on $k'=2^{O(k)}$ variables $\{x_{a_1,\ldots,a_k} : a_1,\ldots,a_k\in\mathbb{F}\}$, and then to constant gap {\sc $k'$-Clique}. In our new reduction, the parameter blow-up throughout is only $2^{O(k)}$, leading to an $n^{\Omega(\log k)}$ lower bound for constant gap {\sc $k$-Clique}.

Since the second step is more complicated, we will briefly introduce the ideas here. Given an {\sc $k$-VectorSum} instance $(V_1,\ldots,V_k,\vec t)$, we build a {\sc CSP} instance on variable set $X=\{x_{a_1,\ldots,a_k} : a_1,\ldots,a_k\in\mathbb{F}\}$. Each variable takes value in $\mathbb F^m$ where $m$ is the dimension specified by the {\sc $k$-VectorSum} problem. In the yes-case, let $\vec{v}_1 \in V_1,\ldots,\vec{v}_k \in V_k$ be a solution that sum up to $\vec t$, we expect $x_{a_1,\ldots,a_k}$ to take the value $\sum_{i\in[k]}a_i\vec{v}_i$. Similar to \cite{Lin21}, we want to make the following three types of tests:
\begin{itemize}
	\item $\forall (a_1,\ldots,a_k),(b_1,\ldots,b_k)\in\mathbb{F}^k$, test whether $x_{a_1,\ldots,a_k}+x_{b_1,\ldots,b_k}=x_{a_1+b_1,\ldots,a_k+b_k}$.
	\item $\forall (a_1,\ldots,a_k) \in \mathbb F^k, a \in \mathbb F$, test whether $x_{a_1,\ldots,a_i+a,\ldots,a_k}-x_{a_1,\ldots,a_k}\in a V_i$. 
	\item $\forall (a_1,\ldots,a_k) \in \mathbb F^k, a \in \mathbb F$, test whether $x_{a_1+a,\ldots,a_k+a}-x_{a_1,\ldots,a_k}=a\vec{t}$.
\end{itemize}

If an assignment passes most of the linearity tests,  then there must be vectors $\vec{u}_1,\ldots,\vec{u}_k\in\mathbb{F}^m$ such that $x_{a_1,\ldots,a_k}=\sum_{i\in[k]}a_i\vec{u}_i$ for most $(a_1,\ldots,a_k)\in\mathbb{F}^k$. The second step is meant to guarantee that the selected vectors indeed come from the input. Finally we need the third step to check whether they sum up to $\vec t$.

Note that in our reduction from {\sc 3SAT} to {\sc $k$-VectorSum}, we require the dimension $m$ to be at least $\Omega(k \log n)$. Thus in the {\sc CSP} instance, we cannot simply leave the alphabet to be $\mathbb F^{m}=n^{\Omega(k)}$, which is too large. To reduce the  dimension, we pick $\ell=\Theta(k+\log n)$ matrices $A_1,\ldots,A_\ell\in\mathbb{F}^{k\times m}$ independently at random, and define a new {\sc CSP} problem on variable set $Y=\{y_{\vec{\alpha},\vec{\beta}}:\vec{\alpha},\vec{\beta}\in \mathbb{F}^k\}$, where each $y_{\vec{\alpha},\vec{\beta}}$ is supposed to take the value
\begin{equation}\label{eq:Yform}
\begin{aligned}
y_{\vec{\alpha},\vec{\beta}} &
=(\vec{\alpha} A_1x_{\vec{\beta}},\ldots,\vec{\alpha} A_\ell x_{\vec{\beta}}) \\ 
& =(\vec{\alpha} A_1\sum_{i\in[k]}\beta_i \vec{v}_i,\ldots,\vec{\alpha} A_\ell \sum_{i\in[k]}\beta_i \vec{v}_i ) \\
& =\sum_{i\in[k],j\in[k]}\beta_i\alpha_j (A_1[j]\vec{v}_i,\ldots,A_\ell[j]\vec{v}_i) \\
& \triangleq \sum_{i\in[k],j\in[k]}\beta_i\alpha_j C_{i,j}.
\end{aligned}
\end{equation}
Now the alphabet size is only $\mathbb F^{\ell}=2^{O(k)}n^{O(1)}$. With this idea in mind, we add local constraints to enforce that the assignment to $Y$ is of the above quadratic form (in terms of $\alpha_1,\ldots,\alpha_k$ and $\beta_1,\ldots,\beta_k$), and then use locally decodable properties of quadratic polynomials to extract information about vectors $\vec v_1,\ldots,\vec v_k$. 

In a high level, our construction generalizes that of \cite{Lin21} by replacing the linear code with the Reed-Muller code based on quadratic polynomials.
\medskip

\noindent\textbf{Expander graph production.} To obtain an FPT time lower bound for {\sc $k$-Clique} with $k^{o(1)}$ gap, we apply the standard expander graph product technique. Starting from a constant gap {\sc $k$-Clique} instance, we amplify the gap using an expander graph $H$ on vertex set $[k]$. The new instance contains $k^t$ groups of vertices. Each group corresponds to a unique path of length-$t$ random walk on $H$, and forms an independent set of size $n^t$. A vertex in a group represents a length-$t$ sequence of vertices from the original instance. Two (sequences of) vertices are linked if and only if the vertices contained in them form a clique in the original instance. By properties of expander graphs, we get a {\sc $k^t$-Clique} instance with gap $(\varepsilon')^t$ for some constant $\varepsilon'$. Take $t=o(\log k)$, we can rule out $k^{o(1)}$-ratio FPT-approximation algorithms for {\sc $k$-Clique} under {\sf ETH}.
\medskip

\noindent \textbf{From gap {\sc $k$-Clique} to {\sf PIH}.} The proof that strong lower bound of constant gap {\sc $k$-Clique} implies {\sf PIH} goes as follows. First, we reduce constant gap {\sc $k$-Clique} to constant gap {\sc $k$-Biclique} in the canonical way. Next, we use a combinatorial object called disperser to amplify the gap from a constant to $\frac{k}{\log k}$. The result then follows from the K{\~o}v{\'a}ri-S{\'o}s-Tur{\'a}n Theorem which states that every $2k$-vertex graph without a $K_{\log k,\log k}$-subgraph has at most $O((2k)^{2-\frac{1}{\log k}})$ edges.

\subsection{Organization of the Paper}
The paper is organized as follows. In Section \ref{sec:pre}, we put some preliminaries, including the definitions of problems, hypotheses, and some algebraic and combinatorial tools used in our proofs. In Section \ref{sec:lowerbound}, we prove the {\sf ETH} lower bound of constant gap {\sc $k$-Clique}. In Section \ref{sec:amplification}, we show how to amplify the gap to rule out $k^{o(1)}$-ratio FPT-approximation algorithms for {\sc $k$-Clique} under {\sf ETH}. In Section \ref{sec:pih}, we show how an almost-tight running time lower bound of constant gap {\sc $k$-Clique} implies {\sf PIH}. Finally, in Section \ref{sec:conclusion}, we conclude with a few open questions.
\section{Preliminaries}
\label{sec:pre}
	\subsection{Problems}

	Here we list all the computational problems which are of relevant to our paper.
	
	\begin{itemize}
		\item {\sc 3SAT.} The input is a 3-CNF formula $\varphi$ with $m$ clauses on $n$ variables. The goal is to decide whether there is a satisfying assignment for $\varphi$.
		
		\item {\sc CSP.} The input of a constraint satisfaction problem is a set of variables $X=\{x_1,\ldots,x_n\}$ together with a family of  constraints $\{C_1,C_2,\ldots,C_m\}$ and an alphabet $\Sigma$.  For every $i\in [m]$,  $C_i=(\vec{s}_i,R_i)$, where $\vec{s}_i=(x_{j_1},\ldots,x_{j_{\ell_i}})$ is an $\ell_i$-tuple of variables for some $\ell_i\in[n]$, and $R_i\subseteq\Sigma^{\ell_i}$ indicates a restriction on valid assignments for those $\ell_i$ variables.  The goal is to find  an assignment $\sigma : X\to \Sigma$  such that for all $i\in[m]$, $\sigma(\vec{s}_i)\in R_i$. We call $n,m,q$ and $|\Sigma|$ respectively the number of vertices, the number of clauses, the arity ($=\max_{i \in [m]} \ell_i$), and the alphabet size of this CSP problem.

		\item {\textsc{$k$-Clique}}. The input is an undirected graph $G=(V_1 \dot\cup\ldots\dot\cup V_k,E)$ with $n$ vertices divided into $k$ disjoint groups. The goal is to decide whether we can pick one vertex from each group, such that they form a clique of size $k$.
		
		\item \textsc{$k$-Biclique}. The input is an undirected bipartite graph $G=(V_1 \dot\cup\ldots\dot\cup V_k, U_1\dot\cup\ldots\dot\cup U_k,E)$, where $n$ vertices are divided into $2k$ disjoint groups. The goal is to decide whether we can pick one vertex from each group, such that they form a biclique $K_{k,k}$.
		
	    \item \textsc{Densest $k$-Subgraph}. The input is an undirected graph $G=(V_1 \dot\cup \ldots\dot\cup V_k,E)$ with $n$ vertices divided into $k$ disjoint groups. The goal is to pick one vertex from each group, such that they induce maximum number of edges.
	    
		\item {\sc $k$-VectorSum.} The input consists of $k$ groups of vectors $V_1,\ldots,V_k \subseteq\mathbb F^{d}$ together with a target vector $\vec t \in \mathbb F^d$, where $\mathbb F$ is a finite field of constant size. The goal is to decide whether there exists $\vec v_1\in V_1,\ldots \vec v_k \in V_k$ such that $\sum_{i=1}^k \vec v_i=\vec t$. Throughout our paper we only need the version that $\vec t$ equals to $\vec 0$, and will omit it afterwards.
	\end{itemize}
	
	\subsection{Hypotheses}
	Now we list some computational complexity hypotheses which are related to our results.
	
	\begin{hypothesis}[Exponential Time Hypothesis ({\sf ETH}) \cite{IP01,IPZ01,Tov84}]
		{\sc 3SAT} with $n$ variables and $m=O(n)$ clauses cannot be solved deterministically in $2^{o(n)}$ time. Moreover, this holds even when restricted to formulae in which each variable appears in at most three clauses.
	\end{hypothesis}
	
	Note that the original statement in \cite{IP01} is does not enforce the requirement that each variable appears in at most three clauses. For the restricted version, we first apply the Sparsification Lemma in \cite{IPZ01}, which implies that without loss of generality we can assume the number of clauses $m = O(n)$. Then we apply Tovey's reduction \cite{Tov84}, which produces a {\sc 3SAT} instance with at most $3m + n = O(n)$ variables and each variable appears in at most three clauses. Thus the restricted version is equivalent to the original statement.
	
	The next hypothesis is {\it Parameterized Inapproximability Hypothesis} ({\sf PIH}), a central conjecture in parameterized complexity. We state it in terms of inapproximability  of {\sc Densest $k$-Subgraph} as follows. 
	
    \begin{hypothesis} [Parameterized Inapproximability Hypothesis ({\sf PIH}) \cite{LRSZ20}]
	    There exists a constant $\varepsilon>0$ such that \textsc{Densest $k$-Subgraph} has no $(1+\varepsilon)$ factor FPT-approximation algorithm. In other words, no algorithm can distinguish the following two cases in $f(k)\cdot n^{O(1)}$ time, for any computable function $f$:
		\begin{itemize}
			\item (Completeness.) There exist $v_1 \in V_1,\ldots, v_k \in V_k$ such that they form a clique.
			\item (Soundness.) For any $v_1 \in V_1,\ldots, v_k \in V_k$, they induce only $\binom{k}{2}/(1+\varepsilon)$ edges.
		\end{itemize}
	\end{hypothesis}
	
	The factor $(1+\varepsilon)$ can be replaced by any constant larger than $1$, and the conjecture remains equivalent. Note that the original statement of {\sf PIH} in \cite{LRSZ20} says that {\sc Densest $k$-Subgraph}  is {\sf W[1]}-hard to approximate, but for our use, we choose a relaxed form which states that it has no constant ratio FPT-approximation algorithm, as in \cite{feldmann2020survey}.

	It is worth noting that the relationship between {\sf PIH} and gap {\sc $k$-Clique} is not completely known yet. If for a graph the number of edges induced by $k$ vertices is only $\approx \varepsilon^2 \binom{k}{2}$, it cannot have a clique of size $>\varepsilon k$. Thus, {\sf PIH} implies  \textsc{$k$-Clique} does not admit constant ratio FPT-approximation algorithms. However, the other direction is not necessarily true (forbidding small clique does not imply low edge density), and it remains an important open problem that whether {\sf PIH} holds if we assume {\sc $k$-Clique} is hard to approximate within any constant factor in FPT time~\cite{feldmann2020survey}.

	\subsection{Low Degree Test}
	Let $\mathbb{F}$ be a field of prime cardinality. We say a function $f$ is $\delta$-close to a function class $\mathcal{F}$ if it is possible to modify at most $\delta$ fraction of values of $f$ such that the modified function lies in $\mathcal{F}$.
	
	The canonical low degree test proposed in \cite{RS96} can query a function $f$ at $d+2$ points, and 
	\begin{itemize}
		\item accepts with probability 1 whenever $f$ is a degree-$d$ polynomial,
		\item rejects with probability at least $\varepsilon>0$ if $f$ is not $\delta$-close to degree-$d$ polynomials, where $\varepsilon,\delta$ are two constants.
	\end{itemize}
	
	Throughout our paper we will consider the function class $\mathcal F$ to be vector-valued degree-$d$ polynomials, namely, \[\mathcal F = \{(f_1, f_2, \ldots, f_{\ell}): \mathbb F^m \to \mathbb F^\ell \mid  \forall i \in [\ell], f_i \text{ is a degree-} d \text{ multivariate polynomial}\}.\] By slightly modifying the proof in \cite{RS96}, the low degree test can be easily generalized to vector-valued version, as formally stated below:
	
	\begin{lemma}\label{LDTVector}
		Let $d < |\mathbb{F}| / 2$ and $m\in \mathbb{N}$. There is an algorithm which, by querying the function $f = (f_1, f_2, \ldots, f_\ell): \mathbb{F}^m\to \mathbb{F}^{\ell}$ at $d+2$ points,
    	\begin{itemize}
    	    \item accepts with probability 1 whenever $f$ lies in $\mathcal F$,
    	    \item rejects with probability at least $\min(\delta / 2, cd^{-2})$ if $f$ is not $\delta$-close to $\mathcal F$, where $c,\delta$ are two constants.
    	\end{itemize}
    	
    	Moreover, the queries are generated by selecting $\vec{x}, \vec{h}\in \mathbb{F}^m$ uniformly at random, and $f$ is queried at $\{\vec{x} + i\vec{h} | 0 \le i < d+2\}$.
	\end{lemma}

    The proof is implicit in literature. To avoid distracting the reader, we defer it to Appendix \ref{sec:ldt}.

	\subsection{Expander Graphs}
	Given a $d$-regular undirected graph $G$ on $n$ vertices, define its \emph{normalized adjacency matrix} to be a matrix $A$ where $A_{ij}$ equals to the number of edges between $(i, j)$ divided by $d$. Define \[\lambda(G) = \max_{\left\|\vec{v}\right\| = 1, \langle \vec{v}, \vec{1}\rangle = 0} \left\| A\vec{v}\right\|_2.\] $G$ is an $(n, d, \lambda)$-expander if and only if $\lambda(G)\le \lambda$, and we have the following two Lemmas.
	
	\begin{lemma}[\cite{AKS87}]\label{ExpanderWalk}
		Let $G$ be an $(n, d, \lambda)$-expander,  $\mathcal{B}\subseteq [n]$ be a set of size $\le \varepsilon n$ for some $0<\varepsilon<1$ and  $(X_1, X_2, \ldots, X_t)$ be a sequence of random variables denoting a length-$t$ random walk where the starting vertex is also picked uniformly at random. Then,
			\[\Pr[\forall 1\le i\le t, X_i\in \mathcal{B}]\le ((1 - \lambda)\sqrt{\varepsilon} + \lambda)^{t-1}.\]
	\end{lemma}
	
	\begin{lemma}[\cite{RVW00}]\label{ExpanderConstruction} 
	    For some constants $d\in \mathbb{N}$, $\lambda < 1$ and for sufficiently large $n$, an $(n,\lambda, d)$-expander can be constructed in $n^{O(1)}$ time.
	\end{lemma}
	
	\subsection{Disperser}

    \begin{definition}[Disperser \cite{CW89, Zuc96a, Zuc96b}]
	    For positive integers $m, k,\ell, r\in \mathbb N$ and constant $\varepsilon\in(0,1)$, an $(m,k,\ell,r,\varepsilon)$-disperser is a collection $\mathcal I$ of $k$ subsets $I_1,\ldots,I_k \subseteq [m]$, each of size $\ell$, such that the union of any $r$ different subsets from the collection has size at least $(1-\varepsilon)m$.
    \end{definition}
    
    Dispersers could be constructed efficiently by probabilistic methods, as in the following Lemma.

    \begin{lemma}
    \label{lem:disperser}
	    For positive integers $m,\ell,r \in \mathbb N$ and constant $\varepsilon \in (0,1)$, let $\ell=\lceil \frac{3m}{\varepsilon r}\rceil$ and let $I_1,\ldots,I_k$ be random $\ell$-subsets of $[m]$. If $\ln k \le \frac{m}{r}$ then $\mathcal I=\{I_1,\ldots,I_k\}$ is an $(m,k,\ell,r,\varepsilon)$-disperser with probability at least $1-e^{-m}$.
    \end{lemma}
    
    For the sake of self-containedness, we put a  proof of this lemma in Appendix \ref{appendix:disperser}. 

\section{An Improved Lower Bound for Constant Gap {\sc $k$-Clique} under {\sf ETH}}
\label{sec:lowerbound}

\subsection{Reduction from {\sc 3SAT} to {\sc $k$-VectorSum}}

To prove Theorem \ref{thm:main}, we first need an $f(k)\cdot n^{\Omega(k)}$-time lower bound for {\sc $k$-VectorSum} under {\sf ETH}. Previously, it is known that {\sc $k$-Clique} has no $f(k)\cdot n^{o(k)}$-time algorithms assuming {\sf ETH}~\cite{ChenHKX06}.  Combining this with the reduction from {\sc $k$-Clique} to {\sc $\Theta(k^2)$-VectorSum}~\cite{abboud2014losing}, we only have an $f(k)\cdot n^{\Omega(\sqrt{k})}$-time lower bound for {\sc $k$-VectorSum} under {\sf ETH}. It is an interesting question whether there is an FPT reduction from {\sc $k$-Clique} to {\sc $k'$-VectorSum} with $k'=O(k)$. In this section, we give a reduction directly from {\sc 3SAT} to {\sc $k$-VectorSum}, which suits our purpose. Recall that in the {\sf ETH} statement we can assume without loss of generality that each variable appears in at most 3 clauses, which is a key ingredient in our proof.
	\begin{theorem}\label{3SATtokVS}
		 There is a reduction which, for every integer $k \in \mathbb N$, and every {\sc 3SAT} formula $\varphi$ with $m$ clauses and $n$ variables such that each variable appears in at most 3 clauses, outputs a {\sc $k$-VectorSum} instance $\Gamma=(\mathbb F,d,V_1,\ldots,V_k)$ with the following properties in $2^{O(n/k)}$ time.
		\begin{itemize}
			\item $\mathbb F=\mathbb F_5$.
			\item $d=O(n)$.
			\item For any $i\in [k]$, distinct $\vec u,\vec v\in V_i$ and any $a\in \mathbb F_5\setminus \{0\}$, $\vec u\ne a\cdot \vec v$.
			\item For any $i \in [k]$, distinct $\vec u,\vec v,\vec w \in V_i$ and any $ a\in \mathbb F_5 \setminus\{0\}, \vec u-\vec w \ne a \cdot (\vec w-\vec v)$.
			\item (Completeness.) If $\varphi$ is satisfiable, then there exists $\vec v_1 \in V_1,\ldots,\vec v_k \in V_k$ such that $\sum_{i=1}^k \vec v_i=\vec 0$.
			\item (Soundness.) If $\varphi$ is not satisfiable, then for any $\vec v_1 \in V_1,\ldots,\vec v_k \in V_k$, $\sum_{i=1}^k\vec v_i \neq \vec 0$.
		\end{itemize}
	\end{theorem}
	\begin{remark}
	    Note if the size of the produced {\sc $k$-VectorSum} instance $N$ appears to be only $2^{o(n/k)}$, we can use brute force to solve it in $N^k=2^{o(n)}$ time, thus solve {\sc 3SAT} in $2^{o(n)}$ time. Therefore, we only need to consider the case $N=2^{\Theta(n/k)}$ without loss of generality, and in this case $d=O(n)=O(k \log N)$.
	\end{remark}

	\begin{proof}[Proof of Theorem \ref{3SATtokVS}]

		Let $C=\mathcal{C}_1 \dot{\cup} \ldots \dot{\cup} \mathcal{C}_k$ be a partition of the clauses into $k$ approximately equal-sized parts. We will let vectors in $V_i$ represent partial satisfying assignments for $\mathcal{C}_i$, and use entries of vectors to check consistency of those partial assignments.
		
		Define $X$ to be the set of variables appearing in exactly two different parts and define $Y$ to be the set of variables appearing in three different parts. Let $d = |X| + 2|Y|$, we associate one entry of vector to each variable $x \in X$ and two entries to each variable $y \in Y$. In the following, we abuse notation a bit and use $\vec v[x]$ to denote the entry in a vector $\vec v \in \mathbb F^d$ associated to a variable $x \in X$, and use $\vec v[y,1],\vec v[y,2]$ to denote the two entries associated to a variable $y \in Y$.
		
		The construction of vector set $V_i$ proceeds as follows. Let $Z_i$ be the set of variables appearing in $\mathcal{C}_i$. For an assignment $\tau: Z_i \to \{0,1\}$ which satisfies all clauses in $\mathcal{C}_i$, we map it to a vector $\vec v \in \mathbb F^d$ in the following way.
		
		Let $x\in Z_i \cap X$ be a variable appearing in $\mathcal{C}_{j_1}$ and $\mathcal{C}_{j_2}$ ($j_1<j_2$),
		\begin{itemize}
			\item in case that $\tau(x) = 0$, set $\vec v[x]=0$. 
			\item in case that $\tau(x) = 1$, set $\vec v[x]=1$ if $i=j_1$, and set $\vec v[x]=-1$ if $i=j_2$.
		\end{itemize}

		Let $y\in Z_i \cap Y$ be a variable appearing in $\mathcal{C}_{j_1},\mathcal{C}_{j_2}$ and $\mathcal{C}_{j_3}$ ($j_1<j_2<j_3$),
		\begin{itemize}
			\item in case that $\tau(y)=0$, set $\vec v[y,1]=0$ and $\vec v[y,2]=0$.
			\item in case that $\tau(y)=1$, set $\vec v[y,1]=1$ and $\vec v[y,2]=1$ if $i=j_1$; set $\vec v[y,1]=-1$ and $\vec v[y,2]=0$ if $i=j_2$; and set $\vec v[y,1]=0$ and $\vec v[y,2]=-1$ if $i=j_3$.
		\end{itemize}
		
		For the remaining entries of $\vec v$ (which are associated to variables in $(X \cup Y) \setminus Z_i$), set them to be 0.
		
		It's easy to see the whole reduction runs in $2^{O(n/k)}$ time, and the dimension $d=O(n)$.

		Now we prove the third and the fourth properties.
		
		For two distinct vectors $\vec u,\vec v \in V_i$, suppose $\vec u[j] \ne \vec v[j]$. It must be the case that one of them is 0 and the other is $\pm 1$. Thus they still differ after being multiplied by any $a\in \mathbb F_5\setminus \{0\}$.
		
		For three distinct vectors $\vec u,\vec v,\vec w \in V_i$, suppose $\vec u[j]\ne \vec w[j]$, then either $\vec v[j]=\vec w[j]$ or $\vec v[j]=\vec u[j]$. In the former case, $\vec u[j]-\vec w[j] \ne 0=\vec w[j]-\vec v[j]$, so they still differ after being multiplied by any $a \in \mathbb F \setminus \{0\}$. In the latter case, suppose $\vec u-\vec w=a \cdot (\vec w-\vec v)$, then $\vec u[j]-\vec w[j]=a \cdot (\vec w[j]-\vec v[j])$ will lead to $a=-1$ and thus $\vec u=\vec v$, a contradiction. Therefore, $\vec u-\vec w \ne a \cdot (\vec w-\vec v)$ for any $a \in \mathbb F_5 \setminus \{0\}$.
		
		Next follows the proof of completeness and soundness.
		
		~
		
		\noindent \textbf{Completeness.} If the {\sc 3SAT} formula $\varphi$ has a satisfying assignment $\tau$, we can pick one vector $\vec v_i$ from each $V_i$ according to the restriction of $\tau$ on $Z_i \cap (X \cup Y)$. Let $\vec v=\sum_{i=1}^k \vec v_i$ be the sum of picked vectors.
		
		For a variable $x \in X$, let $\mathcal{C}_{j_1},\mathcal{C}_{j_2} (j_1<j_2)$ be the two clause parts in which $x$ appears,
			\begin{itemize}
				\item in case that $\tau(x)=0$, $\vec v[x]=0$ since this entry equals to 0 in all vectors.
				\item in case that $\tau(x)=1$, $\vec v[x]=1+(-1)=0$ where 1 comes from $\vec v_{j_1}[x]$ and $-1$ comes from $\vec v_{j_2}[x]$.
			\end{itemize}
			
		For a variable $x \in Y$, let $\mathcal{C}_{j_1},\mathcal{C}_{j_2},\mathcal{C}_{j_3} (j_1<j_2<j_3)$ be the three clause parts in which $x$ appears.
			\begin{itemize}
				\item in case that $\tau(x)=0$, $\vec v[x,1]=\vec v[x,2]=0$ since these entries equal to 0 in all vectors.
				\item in case that $\tau(x)=1$, $\vec v[x,1]=1+(-1)=0$ where 1 comes from $\vec v_{j_1}[x,1]$ and $-1$ comes from $\vec v_{j_2}[x,1]$, and $\vec v[x,2]=1+(-1)=0$ where 1 comes from $\vec v_{j_1}[x,1]$ and $-1$ comes from $\vec v_{j_3}[x,1]$.
			\end{itemize}
		
		~
		
		\noindent \textbf{Soundness.} If the {\sc 3SAT} formula $\varphi$ has no satisfying assignments, any collection of partial assignments satisfying individual clause parts must be inconsistent on some variable in $X \cup Y$. For any $\vec v_1 \in V_1, \ldots, \vec v_k \in V_k$, let $\vec v=\sum_{i=1}^k \vec v_i$.
		
		Suppose assignments for a variable $x \in X$ which appears in $\mathcal{C}_{j_1}$ and $\mathcal{C}_{j_2}$ are inconsistent, there must be one $0$ and one $\pm 1$ in $\vec v_{j_1}[x]$ and $\vec v_{j_2}[x]$. Since this entry equals to 0 in all other vectors, it results that $\vec v[x] \ne 0$.
		
		Suppose assignments for a variable $x \in Y$ which appears in $\mathcal{C}_{j_1},\mathcal{C}_{j_2}$ and $\mathcal{C}_{j_3}$ ($j_1<j_2<j_3$) are inconsistent. If the values for $x$ specified by $\vec v_{j_1}$ and $\vec v_{j_2}$ are inconsistent, there must be one $0$ and one $\pm 1$ in $\vec v_{j_1}[x,1]$ and $\vec v_{j_2}[x,1]$, while in all other vectors this entry equals to 0, thus $\vec v[x,1] \ne 0$.  Otherwise the value for $x$ specified by $\vec v_{j_1}$ and $\vec v_{j_3}$ must be inconsistent, there must be one $0$ and one $\pm 1$ in $\vec v_{j_1}[x,2]$ and $\vec v_{j_3}[x,2]$, while in all other vectors this entry equals to 0, thus $\vec v[x,2] \ne 0$.
		
		Therefore, $\sum_{i=1}^k \vec v_i \ne \vec 0$ as desired.
	\end{proof}

\subsection{Reduction from {\sc $k$-VectorSum} to Constant Gap {\sc $k$-Clique}}

	\begin{theorem}\label{kVStokClique}
		There is an FPT reduction which, given as input a {\sc $k$-VectorSum} instance $\Gamma_0=(\mathbb F,d,V_1,\ldots,V_k)$ with the following properties:
		\begin{itemize}
			\item $\mathbb F=\mathbb F_5$,
			\item $d=O(k\log n)$ where $n=\sum_{i=1}^k |V_i|$ denotes instance size,
			\item for any $i\in [k]$, distinct $\vec u,\vec v\in V_i$ and any $a\in \mathbb F_5\setminus \{0\}$, $\vec u\ne a\cdot \vec v$,
			\item for any $i \in [k]$, distinct $\vec u,\vec v,\vec w \in V_i$ and any $ a\in \mathbb F_5 \setminus\{0\}, \vec u-\vec w \ne a \cdot (\vec w-\vec v)$.
		\end{itemize}
		outputs a {\sc $k'$-Clique} instance $G=(V,E)$ such that
		\begin{itemize}
			\item $k' \le c^k$ for some constant $c$,
			\item (Completeness.) if $\Gamma_0$ is a yes-instance of {\sc $k$-VectorSum}, then $G$ contains a clique of size $k'$,
			\item (Soundness.) if $\Gamma_0$ is a no-instance of {\sc $k$-VectorSum}, then $G$ doesn't contain a clique of size $\varepsilon k'$ for some constant $\varepsilon < 1$.
		\end{itemize}
	\end{theorem}
	
	The first step of the reduction involves $\ell=2k+4 \log n$ matrices $A_1,\ldots,A_{\ell} \in \mathbb F^{k \times d}$. For $\alpha\in \mathbb F^{k},v\in \mathbb F^{d}$, define bilinear function $f(\alpha,v)=(\langle \alpha, A_1v\rangle,\ldots,\langle \alpha,A_{\ell}v\rangle)\in \mathbb F^{\ell}$, where $\langle \cdot, \cdot \rangle$ denotes inner product.
	
	\begin{lemma}[\cite{Lin21}]\label{A0}
		We can find $\ell=2k+4 \log n$ matrices $A_1, A_2, \cdots, A_{\ell}$ in time polynomial in $n,k$, which satisfy the following properties:
		\begin{enumerate}
			\item for any nonzero vector $\vec v \in \mathbb F^d$, there exists $i \in [\ell]$ such that $A_i\vec v \ne \vec 0$,
			\item for any $ i \in [k]$, distinct $ \vec u,\vec v\in V_i$ and nonzero $\alpha \in \mathbb F^k$, $f(\alpha,\vec u) \ne f(\alpha,\vec v)$,
			\item for any $i \in [k]$, distinct $\vec u,\vec v,\vec w \in V_i$ and $\alpha,\alpha'\in \mathbb F^k, f(\alpha,\vec u)+f(\alpha',\vec v) \ne f(\alpha+\alpha',\vec w)$.
		\end{enumerate}
	\end{lemma}
	
	The reduction then goes as follows. For every $\alpha,\beta \in \mathbb F^k$, we introduce a variable $x_{\alpha,\beta}$ which takes value in $\mathbb F^{\ell}$. In the yes-case of {\sc $k$-VectorSum}, there exists $\vec v_1 \in V_1,\ldots,\vec v_k \in V_k$ such that $\sum_{i=1}^k \vec v_i = \vec 0$, and we expect $x_{\alpha,\beta}$ to be $f(\alpha,\sum_{i=1}^k \beta_i\vec v_i)$, in other words, $\sum_{i=1}^k \sum_{j=1}^k \alpha_i\beta_j (A_1[i]\vec v_j,\ldots,A_{\ell}[i]\vec v_j)$ where $A_w[i]$ indicates the $i$-th row of the $w$-th matrix. Note that $f(\alpha,\sum_{i=1}^k \beta_i\vec v_i)$ is a degree-2 polynomial of $\alpha$ and $\beta$.
	
	For simplicity of notation, we will use $e_i\in \mathbb F^k$ to denote the $i$-th unit vector, and use $\mathbf{1}\in \mathbb F^k$ to denote the all-one vector.
	
	We want to apply four types of tests on those variables:
	\begin{enumerate}
		\item Check whether $x:\mathbb F^{2k} \to \mathbb F^\ell$ is a vector-valued degree-2 polynomial. This can be done by the low-degree test described in Lemma \ref{LDTVector}. For each $\vec\alpha,\vec\beta ,\vec t_1,\vec t_2\in \mathbb F^{k}$, check whether $\{x_{\vec\alpha+i\cdot \vec t_1,\vec\beta+i\cdot \vec t_2}|0 \le i \le 3\}$ are point values of a degree-2 polynomial. Each test is applied on 4 variables, and we say the arity of each such test is 4 in shorthand.
		\item Check whether $x$ which maps $(\alpha,\beta)$ to $x_{\alpha,\beta}$ is linear in both $\alpha$ and $\beta$, i.e., whether $x_{\alpha+\alpha',\beta}=x_{\alpha,\beta}+x_{\alpha',\beta},\forall \alpha,\alpha',\beta \in \mathbb F^k$, and $x_{\alpha,\beta+\beta'}=x_{\alpha,\beta}+x_{\alpha,\beta'},\forall \alpha,\beta,\beta' \in \mathbb F^k$. The arity of each such test is 3.
		\item For each $u \in [k],\alpha,\beta \in \mathbb F^k$, check whether
		$x_{\alpha,\beta+e_u}-x_{\alpha,\beta}=f(\alpha,\vec v) \in \mathbb F^{\ell}$ for some $\vec v \in V_u$. The arity of each such test is 2.
		\item For each $\alpha,\beta \in \mathbb F^k$, check whether $x_{\alpha,\beta+\mathbf{1}}-x_{\alpha,\beta}=\vec 0$. The arity of each such test is 2.
	\end{enumerate}
	
	\noindent\textbf{Construction of the Graph.} The vertices are divided into three types.
	Vertices in each type are further partitioned into groups, and each group forms an independent set:
	\begin{description}
		\item[type-1] There are $(|\mathbb F|^{2k})^2$ groups, each of which indicates a test of type 1, and consists of $\le |\mathbb F|^{4\ell}$ vertices corresponding to all satisfying assignments of the 4 variables in the test. 
		
		\item[type-2] There are $2(|\mathbb F|^{k})^3$ groups, each of which indicates a test of type 2, and consists of $\le |\mathbb F|^{2\ell}$ vertices corresponding to all satisfying assignments of the 3 variables in the test.
		
		\item[type-3] There are $|\mathbb F|^{2k}$ groups indexed by $(\alpha,\beta)\in \mathbb F^{2k}$, each consisting of $\mathbb F^{\ell}$ vertices which correspond to $\mathbb F^{\ell}$ possible assignments for variable $x_{\alpha,\beta}$.
	\end{description}
	
	We make copies of vertices, so that the numbers of type-1 groups and type-2 groups are the same, and their sum equals to the number of type-3 groups. Specifically, the three types of vertices are made into $2,|\mathbb F|^{k},4\mathbb F^{2k}$ copies, respectively. The total number of groups is therefore $k'=8|\mathbb F|^{4k}$, while the total number of vertices is at most $2|\mathbb F|^{4k+4\ell}+2|\mathbb F|^{4k+2\ell}+4|\mathbb F|^{4k+\ell}=|\mathbb F|^{O(k+\log n)}$.
	
	The edges are specified as follows:
	\begin{enumerate}
		\item A variable (type-3) vertex and a test (type-1/2) vertex are linked if and only if they specify the same assignment for the variable, or the test is irrelevant of that variable.
		\item Two test vertices are linked if and only if they are consistent in all variables which appear in both tests. 
		\item Two variable vertices are linked if and only if the assignments specified by them can pass the above-mentioned third and fourth tests, or there is no such a test between them.
	\end{enumerate}
	
	Two different copies of a same vertex are always linked.
	
	~
	
	\noindent\textbf{Proof of Completeness.} If $\Gamma_0$ is a yes-instance of {\sc $k$-VectorSum}, i.e., there exists $\vec v_1 \in V_1,\ldots, \vec v_k \in V_k$ such that $\sum_{i=1}^k \vec v_i=\vec 0$, by letting $x_{\alpha,\beta}$ take value $f(\alpha,\sum_{i=1}^k \beta_i \vec v_i)$, it's easy to see that such an assignment can pass all tests. Therefore, by picking a vertex from each group accordingly, one can obtain a clique of size $k'$.
	
	~
	
	\noindent\textbf{Proof of Soundness.} If $\Gamma_0$ is a no-instance of {\sc $k$-VectorSum}, we will prove that there is no clique of size $\ge (1-\varepsilon)k'$ in $G$ for some  small constant $\varepsilon$.
	
	Prove by contradiction. If there is a clique of size $\ge (1-\varepsilon)k'$, it must contain vertices from $\ge (1-2\varepsilon)$ fraction of type-3 groups which represent variables, vertices from $\ge (1-4\varepsilon)$ fraction of type-1 groups which represent low-degree tests, vertices from $\ge (1-8\varepsilon)$ fraction of type-2 groups which represent linearity tests $x_{\alpha+\alpha',\beta}=x_{\alpha,\beta}+x_{\alpha',\beta}$ and vertices from $\ge (1-8\varepsilon)$ fraction of type-2 groups which represent linearity tests $x_{\alpha,\beta+\beta'}=x_{\alpha,\beta}+x_{\alpha,\beta'}$.
	
	In the following, we will denote by $\overline x_{\alpha,\beta}$ the assignment for $x_{\alpha,\beta}$ specified by the clique. If no assignment for $x_{\alpha,\beta}$ is specified, set $\overline x_{\alpha ,\beta}$ arbitrarily as long as it is consistent with all selected test vertices (it is always possible since the selected test vertices are themselves consistent). As almost all low-degree tests and linearity tests are passed, we have:
	
	\begin{lemma}\label{B1}
		If $\varepsilon<\frac{c}{16}$ where $c$ is the constant in Lemma \ref{LDTVector}, and there is a clique of size $\ge (1-\varepsilon)k'$, then
		 the function $\pi(\alpha,\beta)=\overline x_{\alpha,\beta}$ ($\alpha,\beta\in\mathbb{F}^k$)
		is $9\varepsilon$-close to a function on $\alpha,\beta$ of the form 	\[\sum_{i=1}^k\sum_{j=1}^k\alpha_i\beta_jC_{i,j}\] where $C_{i,j}\in \mathbb F^l$ denotes the coefficient of the term $\alpha_i\beta_j$.
	\end{lemma}
	\begin{proof}
		Plugging $\delta=9\varepsilon$ into Theorem \ref{LDTVector}, if $\pi(\alpha,\beta)$ is
		not $\delta$-close to any degree-2 polynomial, at least $\min(\delta/2,cd^{-2})>4\varepsilon$ fraction of the degree-2 polynomial tests will not be passed. However, when there is a clique of size $\ge (1-\varepsilon)k'$, only $\le 4\varepsilon$ fraction of degree-2 polynomial tests may fail. Therefore, $\pi$ must be $9\varepsilon$-close to a function of the form
		\begin{equation}
		\sum_{i=1}^k\sum_{j=1}^k\alpha_i\alpha_jA_{i,j}+\sum_{i=1}^k\sum_{j=1}^k\beta_i\beta_jB_{i,j}+\sum_{i=1}^k\sum_{j=1}^k\alpha_i\beta_jC_{i,j}+\sum_{i=1}^kD_i \alpha_i+\sum_{i=1}^kE_i\beta_i+F\label{quadeq}
		\end{equation}
		where each coefficient is in $\mathbb F^l$. We only need to prove that if they also pass most of the linearity tests on $\alpha$ and on $\beta$, all coefficients except $C_{i,j}$ must be zeros.
		
		Suppose $A_{i,j} \ne 0$ for some $i,j \in [k]$. Regarding $x_{\alpha+\alpha',\beta}-x_{\alpha,\beta}-x_{\alpha',\beta}$ as a function on $3k$ variables $\alpha_1\ldots \alpha_k,\alpha'_1\ldots \alpha'_k,\beta_1 \ldots \beta_k$ and expand it by \eqref{quadeq}. There is a term $A_{i,j}(\alpha_{i}\alpha_{j}'+\alpha_i'\alpha_j)$ which can never be canceled. The function $x_{\alpha+\alpha',\beta}-x_{\alpha,\beta}-x_{\alpha',\beta}$ is not a zero function and by Schwartz-Zippel Lemma, only $\frac{2}{|\mathbb F|}$ fraction of $\alpha,\alpha',\beta$ can make it equal to zero. However, by union bound, there are at least $1-8\varepsilon-3\cdot(9\varepsilon)>\frac{2}{|\mathbb F|}$ fraction of $(\alpha,\alpha',\beta)$  such that $\overline x_{\alpha,\beta},\overline x_{\alpha',\beta},\overline x_{\alpha+\alpha',\beta}$ are all specified value according to the clique, consistent with equation \eqref{quadeq}, and satisfying $\overline x_{\alpha,\beta}+\overline x_{\alpha',\beta}=\overline x_{\alpha+\alpha',\beta}$, a contradiction. Therefore, $\forall i,j \in [k], A_{i,j}=0$.
		
		Similar arguments can be applied for the other coefficients and are omitted here. The only term remaining is $\sum_{i=1}^k\sum_{j=1}^k\alpha_i\beta_jC_{i,j}$ as desired.
	\end{proof}
	
	We call a variable $x_{\alpha,\beta}$ \textit{good} if the clique consists of a variable vertex of it and it satisfies $\overline x_{\alpha,\beta}=\sum_{i=1}^k\sum_{j=1}^k\alpha_i\beta_jC_{i,j}$. Let $\varepsilon'=2\varepsilon+9\varepsilon$, from the above we know at least $(1-\varepsilon')$ fraction of variables are good. Recall that from the construction of our graph and the property of clique, all arity-2 constraints between good variables are satisfied.
	
	We call an $\alpha \in \mathbb F^k$ \textit{excellent} if $\Pr_{\beta\in_R\mathbb F^k}[x_{\alpha,\beta} \text{ is good}]\ge \frac{2}{3}$. By Markov's Inequality, at least $1-3\varepsilon'$ fraction of $\alpha$'s are excellent.

	\begin{lemma}\label{B2}
		For each excellent $\alpha$ and for each $u \in [k]$, $\sum_{i=1}^k \alpha_i C_{i,u}=f(\alpha,\vec v)$ for some unique $\vec v \in V_u$.
	\end{lemma}
	\begin{proof}
	
		For any fixed $u \in [k]$ and $\alpha$, the set of edges between the vertex of $x_{\alpha,\beta}$ and the vertex of $x_{\alpha,\beta+e_u}$ for all $\beta$ can be partitioned into disjoint  length-5 cycles $\{x_{\alpha,\beta},x_{\alpha,\beta+e_u},\ldots,x_{\alpha,\beta+4e_u}\}$ since the characteristic of $\mathbb F$ is 5. Observe that if at most two variables in a 5-cycle are not good, there still exist two adjacent vertices that are good.
		
		For an excellent $\alpha$, at most $\frac{1}{3}\le\frac{2}{5}$ fraction of variables $x_{\alpha,\beta}$ are not good by definition, so there must exists a $\beta\in\mathbb{F}^k$ such that  $x_{\alpha,\beta}$ and $x_{\alpha,\beta+e_u}$ are both good variables. According to the definition of good variables, we have \[\overline x_{\alpha,\beta+e_u}-\overline x_{\alpha,\beta} = \sum_{i=1}^k \alpha_iC_{i,u}\] and from the third type of constraints between them we can infer that \[\overline x_{\alpha,\beta+e_u}-\overline x_{\alpha,\beta}=f(\alpha,\vec v)\] for some $\vec v \in V_u$.
		
		Additionally, since for $\vec v \in V_u$, $f(\alpha,\vec v)$ are all different (the second property in Lemma \ref{A0}), for an excellent $\alpha$ and for all $u \in [k]$,  we can deduce $\sum_{i=1}^k \alpha_i C_{i,u}=f(\alpha,\vec v)$ for some unique $\vec v \in V_u$. 
	\end{proof}
	
	In the following when $u\in [k]$ is fixed and omitted, we use $\vec v_{\alpha}$ to denote the unique vector in $V_u$ specified by an excellent $\alpha$. For an $\alpha$ which is not excellent, we also assign a unique vector $\vec v_{\alpha} \in V_u$ to it arbitrarily so that $v_{\alpha}$ is defined for all $\alpha \in \mathbb F^k$.
	
	\begin{lemma}\label{B3}
		If there is a clique of size $\ge (1-\varepsilon)k'$, then for each $u \in [k]$, \[C_{i,u}=(A_1[i]\vec v,\ldots, A_{\ell}[i]\vec v)\] for some unique $\vec v \in V_u$, where $A_w[i]$ indicates the $i$-th row of the $w$-th matrix.
	\end{lemma}
	\begin{proof}
		Fix any $u \in [k]$, we first argue that at least $\ge \frac{3}{10}$ fraction of $\alpha$ specify the same $\vec v \in V_u$. It suffices to prove \[\Pr_{\alpha,\alpha'\in_R \mathbb F^k}[\vec v_{\alpha}=\vec v_{\alpha'}]\ge \frac{3}{10}.\]
		
		If this is not true, by union bound $3\cdot\frac{3}{10}+3(3\varepsilon')<1$, there must exist $\alpha,\alpha'$ such that the following two conditions hold:
		\begin{enumerate}
			\item $\alpha,\alpha',\alpha+\alpha'$ are all excellent.
			\item $\vec v_{\alpha},\vec v_{\alpha'},\vec v_{\alpha+\alpha'}$ are all different.
		\end{enumerate}
	
		Now that they are all excellent, we have
		\begin{align*}
		    f(\alpha+\alpha',\vec v_{\alpha+\alpha'}) & =\sum_{i=1}^k (\alpha_i+\alpha'_i)C_{i,u}\\ 
		& = f(\alpha,\vec v_{\alpha})+f(\alpha',\vec v_{\alpha'}),
		\end{align*}
		contradicting with the third property in Lemma \ref{A0}.
		
		Therefore, at least $\frac{3}{10}-3\varepsilon'>\frac{1}{|\mathbb F|}$ fraction of $\alpha$ are all excellent and specify the same $\vec v \in V_u$. Now suppose $C_{i,u} \ne (A_1[i]\vec v,\ldots, A_{\ell}[i]\vec v)$, then there are at most $\frac{1}{|\mathbb F|}$ fraction of $\alpha \in \mathbb F^k$ making $\sum_{i=1}^k \alpha_i C_{i,u}=f(\alpha,\vec v)$ by Schwartz-Zippel Lemma. However, we have $>\frac{1}{|\mathbb F|}$ fraction of such $\alpha$, a contradiction.
	\end{proof}
	
	\begin{lemma}\label{B4}
		If there is a clique of size $\ge (1-\varepsilon)k'$, then there exists $\vec v_1 \in V_1,...,\vec v_k \in V_k$ such that $\sum_{i=1}^k \vec v_i=\vec 0$.
	\end{lemma}
	
	\begin{proof}
		From Lemma \ref{B3} we know that for each $u \in [k]$, $C_{i,u}=(A_1[i]\vec v,\ldots, A_{\ell}[i]\vec v)$ for some unique $\vec v \in V_u$. Thus $\overline x_{\alpha,\beta}$ indeed equals to $f(\alpha,\sum_{i=1}^k\beta_i\vec v_i)$ for every good variable $x_{\alpha,\beta}$. 
		
		The remaining proof is very similar to the proof of Lemma \ref{B2}. Edges between the vertex of $x_{\alpha,\beta}$ and the vertex of $x_{\alpha,\beta+\mathbf{1}}$ for all $\beta\in\mathbb{F}^k$ can be divided into disjoint length-5 cycles $\{x_{\alpha,\beta},x_{\alpha,\beta+\mathbf{1}},\ldots,x_{\alpha,\beta+4\cdot\mathbf{1}}\}$ since the characteristic of $\mathbb F$ is 5.
		
		For an excellent $\alpha$, at most $\frac{1}{3}\le\frac{2}{5}$ fraction of variables $x_{\alpha,\beta}$ are not good, so there must be two variables $x_{\alpha,\beta}$ and $x_{\alpha,\beta+\mathbf{1}}$ which are both good. According to the definition of good variables, we have \[\overline x_{\alpha,\beta+s}-\overline x_{\alpha,\beta}=f\left(\alpha,\sum_{i=1}^k \vec v_i\right)\] and from the fourth type of constraints between them we can infer that \[\overline x_{\alpha,\beta+\mathbf{1}} - \overline x_{\alpha,\beta} = \vec 0.\]
		
		Therefore, $f(\alpha,\sum_{i=1}^k\vec v_i)=\vec 0$ for every excellent $\alpha$.
		
		Suppose $\sum_{i=1}^k \vec v_i \ne \vec 0$, then there exists $i \in [\ell]$ such that $A_i(\sum_{i=1}^k \vec v_i)\ne \vec 0$ by the first property in Lemma \ref{A0}. There are at most $\frac{1}{|\mathbb F|}$ fraction of $\alpha$ such that $f(\alpha,\sum_{i=1}^k\vec v_i)=\vec 0$ by Schwartz-Zippel Lemma, but we have $\ge 1-3\varepsilon'$ fraction of excellent $\alpha$, a contradiction.
	\end{proof}
	
	\subsection{Putting Things Together}
	
	Combing Theorem \ref{3SATtokVS} with Theorem \ref{kVStokClique}, we obtain an improved lower bound for constant gap {\sc $k$-Clique} under {\sf ETH} as follows.
	
	\begin{restated_thm}{thm:main}
		 Assuming {\sf ETH}, for any constant $c>1$ and any computable function $f$, no algorithm can find a clique of size $k/c$ in the {\sc $k$-Clique} problem in $f(k)n^{o(\log k)}$ time.
	\end{restated_thm}
	\begin{proof}
	    Without loss of generality assume $f$ is non-decreasing and unbounded. Given a {\sc 3SAT} formula $\varphi$ with $n$ variables, each appearing in at most 3 clauses, we first run the reduction in Theorem \ref{3SATtokVS} to produce a {\sc $k$-VectorSum} instance of size $2^{O(n/k)}$, then run the reduction in Theorem \ref{kVStokClique} to produce a constant gap {\sc $c^k$-Clique} instance of size at most $c^k 2^{O(n/k)}$. Using the graph product method, we can amplify the gap to any constant, while keeping the parameter $k' = c^{O(k)}$ and instance size $n'\le c^{O(k)}2^{O(n/k)}$. Therefore, an $f(k')n^{o(\log k')}$ time algorithm for constant gap {\sc $k'$-Clique} would lead to an algorithm for {\sc 3SAT} in  \[f(k')(n')^{o(\log k')}\le f(c^{O(k)}) (c^{O(k)} 2^{O(n/k)})^{o(k)} \le 2^{o(n)}\] time, contradicting {\sf ETH}. The last inequality holds because $n$ can be sufficiently large compared to $k$.
	\end{proof}
	
Below we present two remarks about possible extensions of our results on {\sf ETH} lower bounds of gap {\sc $k$-Clique}.
	
\begin{remark}[On higher degree Reed-Muller Codes]
    It is natural to extend our idea to obtain a reduction from {\sc $k$-VectorSum} to a CSP problem with $<2^{o(k)}$ variables using Reed–Muller code with larger degree polynomials. However, the reduction from CSP to {\sc $k$-Clique} has such an important property: when there is a clique of size $\varepsilon k$ for some constant $\varepsilon$, the following two conditions hold:
\begin{enumerate}
    \item A constant fraction of arity-$d$ constraints $(d>2)$ are satisfied.
    \item All arity-$2$ constraints between a constant fraction of variables are satisfied.
\end{enumerate}
The second condition holds because an arity-2 constraint between two variables can be directly transformed into an edge between two vertices. If there is a large clique, it means all arity-2 constraints between the involved variables are simultaneously satisfied. However, if we use larger degree polynomials, the arity of constraints has to be larger, too. It cannot directly fit into the framework of {\sc $k$-Clique}. If this barrier can be broken, it may be possible to obtain reductions with an even smaller parameter blow-up using larger degree polynomials.
\end{remark}

\begin{remark}[On locally decodable codes]
Our reduction implicitly depends on the property of $2$-query locally decodable code, that we could decode $f(\alpha,\vec{v_i})$ for some fixed $\alpha$ by querying only $2$ positions. As pointed out in \cite{GKST06}, $2$-query locally decodable code has at least an exponential blow up. Hence our method is optimal in this sense. We could also consider how to remove this dependence.
\end{remark}
\section{$k^{o(1)}$-Ratio FPT Inapproximability  of {\sc $k$-Clique} under {\sf ETH}}
\label{sec:amplification}

    In this section, we show how to use expander graphs to amplify the gap efficiently, and how it leads to an improved inapproximability ratio of {\sc $k$-Clique} in FPT time under {\sf ETH}. The idea comes from the classical technique used to amplify gap in the non-parameterized version of {\sc Clique} problem, which was proposed by Alon et al. \cite{AFWZ95}.

	\begin{theorem}\label{kCliqueGapAmp}
        For some constants $d \in \mathbb N, 0<\lambda<1$ and for any $t \in \mathbb N$, there is an algorithm which runs in $O(k^2d^{2t}|V|^{2t})$ time, on input an instance $\Gamma=(V,E)$ of {\sc $k$-Clique} problem, outputs an instance $\Gamma'$ of {\sc $k'$-Clique} problem such that
        \begin{itemize}
            \item $k'=kd^{t-1}$.
            \item (Completeness.) If $\Gamma$ has a $k$-clique, then $\Gamma'$ has a $k'$-clique.
            \item (Soundness.) If $\Gamma$ has no $\varepsilon k$-clique, then $\Gamma'$ has no clique of size $k'((1 - \lambda)\sqrt{\varepsilon} + \lambda)^{t-1}$.
        \end{itemize}
    \end{theorem}

	\begin{proof}
	Let $H$ be an $(k, \lambda, d)$-expander constructed from Lemma \ref{ExpanderConstruction}. We construct $\Gamma'$ as follows. Each of the $k'=kd^{t-1}$ groups in $\Gamma'$ is associated with a unique path of length-$t$ random walk on $H$. We use $(c_1,\ldots,c_t)$ to name a group in $\Gamma'$, where each $c_i \in [k]$ indicates a group in $\Gamma$.
	
	A vertex in $\Gamma'$ is a length-$t$ sequence of vertices in $\Gamma$. Namely, there is a vertex $(u_1,\ldots,u_t)$ in the $(c_1,\ldots,c_t)$-th group in $\Gamma'$ if and only if each $u_i$ is belongs to group $c_i$ in $\Gamma$. Therefore, the total number of vertices is at most $kd^{t-1}|V|^t$ in $\Gamma'$.
	
	A vertex $(u_1, \ldots, u_t)$ in group $(c_1, \ldots, c_t)$ is linked to a vertex $(v_1, \ldots, v_t)$ in group $(d_1, \ldots, d_t)$ if and only if
	\begin{itemize}
		\item $(c_1, \ldots, c_t) \ne (d_1, \ldots, d_t)$,
		
		\item and the vertices $\{v_1, \ldots, v_t, u_1, \ldots, u_t\}$ form a clique in $\Gamma$.
	\end{itemize}
	
	The reduction runs in $O(k^2d^{2t}|V|^{2t})$ time by simply  enumerating every pair of vertices in $\Gamma'$ and checking if there is an edge between them.
	
	~
	
	\noindent\textbf{Completeness.}
	Let $\{v_1,\ldots, v_k\}$ be an $k$-clique in $\Gamma$. Then in group $(c_1, \ldots, c_t)$ in $\Gamma'$, we can pick the vertex $(v_{c_1}, \ldots, v_{c_t})$. It's easy to see those vertices form an $kd^{t-1}$-clique.
	
	~
	
	\noindent\textbf{Soundness.}
	For any clique $V$ in $\Gamma'$, let $U$ be the collection of vertices in $\Gamma$ which appear as part of the name of a vertex in $V$. Since $V$ is a clique in $\Gamma'$, it follows by construction that $U$ is also a clique in $\Gamma$ and thus $|U| \le \varepsilon k$. Recall that each $ (c_1,\ldots, c_t)$ represents a length-$t$ random walk on $H$, and all those $c_i$'s lie in a set of size $\le \varepsilon k$ (which corresponds to the collection of groups that vertices in $U$ belong to). By plugging $n=k, |\mathcal B|\le \varepsilon k$ into Lemma \ref{ExpanderWalk}, the number of different groups that vertices in $V$ belong to is bounded by $kd^{t-1}((1-\lambda)\sqrt{\varepsilon} + \lambda)^{t-1}$, and so is $|V|$.
	\end{proof}
	
	For any function $\delta(k)=o(1)$, by setting $t$ to be as large as some $o(\log k)$, we can make $\varepsilon'=((1-\lambda)\sqrt{\varepsilon} + \lambda)^{t-1}$ smaller than $k^{-\delta(k)}$ while keeping $k'=kd^{t-1} \le k^{O(1)}$. Thus, by combining Theorem \ref{thm:main} and Theorem \ref{kCliqueGapAmp}, we have the following corollary.
	\begin{restated_cor}{cor:main}
		Assuming {\sf ETH}, for any $g(k) = k^{o(1)}$, the {\sc $k$-Clique} problem has no $g(k)$-FPT-approximation algorithm.
	\end{restated_cor}
\section{From Constant Gap {\sc $k$-Clique} to {\sf PIH}}
\label{sec:pih}

In this section we will show that strong lower bound of constant gap \textsc{$k$-Clique} implies {\sf PIH}. For simplicity of notation, we additionally define problems {\sc Gap-clique($k$, $\ell$)} and {\sc Gap-biclique($k$, $\ell$)} ($k>\ell$), whose definitions are almost the same as {\sc $k$-Clique} and {\sc $k$-Biclique}, except that the soundness parameter is $\ell$. We have the following theorem:

\begin{restated_thm}{thm:gapcliqe2pih}
    If {\sc Gap-clique($k$, $\varepsilon k$)} does not admit $f(k) \cdot n^{O\left(\frac{k}{\log k}\right)}$-time algorithms for some $0 < \varepsilon < 1$, then {\sf PIH} is true.
\end{restated_thm}

The proof is relatively elementary and consists of three steps: first reduce constant gap {\sc $k$-Clique} to constant gap {\sc $k$-Biclique}, then use a \textit{disperser} to compress the soundness parameter to $\frac{\log k}{k}$, finally use the K\H ov\' ari-S\'os-Tur\'an Theorem to show that in the soundness case, the density of every $2k$-vertex bipartite subgraph is low.

\begin{lemma}\label{lem:clique2biclique}
    If {\sc Gap-clique($k$, $\varepsilon k$)} does not admit $f(k) \cdot n^{O\left(\frac{k}{\log k}\right)}$-time algorithms for some $0 < \varepsilon < 1$, then neither does {\sc Gap-biclique($k$, $\frac{1 + \varepsilon}{2}k$)}.
\end{lemma}
\begin{proof}
    We reduce a {\sc Gap-clique($k$, $\varepsilon k$)} instance $G=(V_1\dot\cup\ldots\dot\cup V_k,E)$ to a {\sc Gap-biclique($k$, $\frac{1 + \varepsilon}{2}k$)} instance  $G'=(U_1\dot\cup\ldots\dot\cup U_k, W_1\dot\cup\ldots\dot\cup W_k, E)$ as follows.
	
	The vertex sets in each side are just copies of $V$, i.e., $U_i=W_i=V_i,\forall i \in [k]$. Two vertices $u \in U_i, w \in W_j$ where $i\ne j$ are linked if and only if their corresponding vertices are linked in $G$, while two vertices $u \in U_i, w \in W_i$ are linked iff they correspond to the same vertex in $G$.
	
	The completeness case is obvious. In the soundness case, suppose we can pick $\frac{1 + \varepsilon}{2}k$ vertices from different parts of $U$ and $\frac{1 + \varepsilon}{2}k$ vertices from different parts of $W$ such that they form a biclique. Let the collection of picked vertices be $S$. Then there must be an index set $\mathcal I$ of size $\varepsilon k$ such that $\forall i \in \mathcal I$, $(S \cap U_i \ne \emptyset) \land (S \cap W_i \ne \emptyset)$. Moreover, $|S \cap U_i|=|S \cap W_i|=1$ by our construction of edges between $U_i$ and $W_i$. Then consider the set $\bigcup_{i \in \mathcal I} (S \cap U_i)$ which is of size at least $\varepsilon k$. The vertices in it must form a clique of size $|\mathcal I|$ in the original graph $G$.
\end{proof}

\begin{lemma}
\label{lem:logk}
    If {\sc Gap-biclique($k$, $\varepsilon k$)} does not admit $f(k) \cdot n^{O\left(\frac{k}{\log k}\right)}$-time algorithms for some constant $0 < \varepsilon < 1$, then for any constant $0<c<1$, no algorithm can solve {\sc Gap-biclique($k$, $c\log k$)} in $f(k) \cdot n^{O(1)}$ time.
\end{lemma}
\begin{proof}
	Given a {\sc Gap-biclique($k$, $\varepsilon k$)} instance $G=(U_1 \dot\cup\ldots\dot\cup U_k,W_1 \dot\cup\ldots\dot\cup W_k, E)$, let $\ell = \lceil \frac{3k}{\varepsilon c\log k}\rceil$ and let $\mathcal I=(I_1,\ldots,I_k)$ be a $(k,k,\ell,c\log k,1-\varepsilon)$-disperser. Since the size of $\mathcal{I}$ is independent of $n$, we can deterministically enumerate all possible $\mathcal{I}$ to find a valid one in $f(k)$ time. The existence of such a disperser is guaranteed by Lemma \ref{lem:disperser}. We construct a new {\sc Gap-biclique($k$, $c\log k$)} instance $G'=(U'_1 \dot\cup\ldots\dot\cup U'_k,W'_1 \dot\cup\ldots\dot\cup W'_k, E)$ as follows.
	
	The groups of vertices in $G'$ correspond to the combination of groups in $G$ according to the disperser. Specifically, for $1 \le i \le k$, let $I_i=\{i_1,\ldots, i_\ell\}$, then each vertex in $U'_i$ will correspond to a tuple of vertices $(u_{i_1},\ldots, u_{i_\ell})$ where $u_{i_j}$ comes from $U_{i_j}$ in $G$ for all $1 \le j \le \ell$. The construction of right vertices $W'_1 ,\ldots W'_k$ is similar. The size of the new instance is therefore at most $n^{\ell}= n^{O(\frac{k}{\log k})}$.
	
	An edge between a left vertex $(u_{i_1},\ldots,u_{i_\ell})$ and a right vertex $(w_{j_1},\ldots,w_{j_\ell})$ exists if and only if the vertices $\{u_{i_1},\ldots,u_{i_\ell}, w_{j_1},\ldots, w_{j_\ell}\}$ form a biclique $K_{\ell,\ell}$ in $G$.
	
	The completeness case is still obvious, and we focus on the soundness case. Prove by contradiction, if there exists $c\log k$ vertices from different groups of $U'$ and $c\log k$ vertices from different groups of $W'$ which form a biclique $K_{c\log k,c\log k}$, let $S$ be the collection of vertices which appear as part of one of the $2c\log k$ tuples. For $1 \le i \le k$, arbitrarily pick one vertex from each $S \cap U_i$, $S \cap W_i$ if not empty, then we claim that the resulting collection must be a biclique of size $\ge \varepsilon k$ on both sides. The promise of biclique is from our construction, while the size is guaranteed by properties of the disperser.
	
	Therefore, an $f(k)\cdot n^{O(1)}$ time algorithm for the {\sc Gap-biclique($k$, $c\log k$)} problem would lead to an $f(k) \cdot \left(n^{O\left(\frac{k}{\log k}\right)}\right)^{O(1)} = f(k) \cdot n^{O\left(\frac{k}{\log k}\right)}$ time algorithm for {\sc Gap-biclique($k$, $\varepsilon k$)} problem. This completes the proof.
\end{proof}

\begin{theorem}(K{\~o}v{\'a}ri-S{\'o}s-Tur{\'a}n, \cite{Kvri1954OnAP})
\label{thm:turan}
For any graph $G$ on $n$ vertices and any $a \ge 2$, if $G$ does not contain $K_{a,a}$ as a subgraph, then $G$ has at most $\frac{1}{2}(a-1)^{\frac{1}{a}}n^{2-\frac{1}{a}}+\frac{1}{2}(a-1)n$ edges.
\end{theorem}

\begin{proof}[Proof of Theorem~\ref{thm:gapcliqe2pih}]

By plugging in $n=k, a = \frac{1}{2}\log k$ in Theorem \ref{thm:turan}, for sufficiently large $k$, $\frac{1}{2}(a-1)^{\frac{1}{a}}n^{2-\frac{1}{a}}+\frac{1}{2}(a-1)n$ is no more than $\varepsilon' k^2$ for some constant $0<\varepsilon'<1$. Theorem \ref{thm:turan} and  Lemma~\ref{lem:logk} imply that if {\sc Gap-biclique($k$, $\varepsilon k$)} does not admit $f(k) \cdot n^{O\left(\frac{k}{\log k}\right)}$-time algorithms, then no FPT algorithm can distinguish the following cases for a \textsc{$k$-Biclique} instance $G=(U,W,E)$ with
\begin{itemize}
    \item \textit{(Completeness.)} there exists $u_1 \in U_1,\ldots,u_k \in U_k, w_1 \in W_1, \ldots w_k \in W_k$ such that they form a biclique $K_{k,k}$.
    \item \textit{(Soundness.)} for all $u_1 \in U_1,\ldots,u_k \in U_k, w_1 \in W_1, \ldots w_k \in W_k$, the vertex set $\{u_1,\ldots,u_k,v_1,\ldots,v_k\}$ induce a subgraph with at most $\varepsilon' k^2$ edges for some constant $0<\varepsilon'<1$.
\end{itemize}
We link all pairs of vertices which are on the same side but not in the same group. In the completeness case, we can find $2k$ vertices from distinct groups such that they form a clique and thus the number of edges induced is $\binom{2k}{2}$. In the soundness case, the number of edges induced by $2k$ vertices from distinct groups is at most $\varepsilon' k^2 + 2\binom{k}{2} < \varepsilon'' \binom{2k}{2}$ for some $0<\varepsilon''<1$.

At last, by Lemma~\ref{lem:clique2biclique}, if {\sc Gap-clique($k$, $\varepsilon k$)} does not admit $f(k) \cdot n^{O\left(\frac{k}{\log k}\right)}$-time algorithms for some $0 < \varepsilon < 1$, then neither does {\sc Gap-biclique($k$, $\frac{1 + \varepsilon}{2}k$)}, hence no FPT-algorithm can approximate \textsc{Densest $k$-Subgraph} to an $\varepsilon''$ factor.
\end{proof}

\section{Conclusion}
\label{sec:conclusion}

In this paper,  we provide a tighter {\sf ETH}-lower bound for constant gap \textsc{$k$-Clique} by replacing the Hardamard code used in~\cite{Lin21} by the Reed-Muller Code with degree-2 polynomials. We use gap amplification techniques by expander graphs to rule out $k^{o(1)}$-ratio FPT-approximation algorithms for {\sc $k$-Clique} under {\sf ETH}. We also study the relationship between the constant gap \textsc{$k$-Clique} problem and \textsf{PIH}. We show that almost tight lower bounds for constant gap \textsc{$k$-Clique} can imply \textsf{PIH}.

A natural open question is whether we can derive such a lower bound for constant gap {\sc $k$-Clique} under {\sf ETH}. Formally, it is stated as follows:

\noindent\textbf{Question 1.} Assuming {\sf ETH}, does constant gap {\sc $k$-Clique} admit an algorithm in $f(k)\cdot n^{O(\frac{k}{\log k})}$ time?

It is also worth noting that assuming {\sf ETH}, there is no $2^{o(n)}$-time algorithm for non-parameterized
{\sc Max-Clique} problem on $n$-vertex graphs. 
Hence, it is natural to ask whether our technique can be
analogously applied to non-parameterized {\sc Max-Clique} to obtain tight lower bounds:

\noindent\textbf{Question 2.} Assuming {\sf ETH}, does constant gap {\sc Max-Clique} admit an algorithm in $2^{o(n)}$ time?

\bibliographystyle{alpha}
\bibliography{main}

\appendix
\section{Vector-valued Low-degree Test}
\label{sec:ldt}

In this subsection, we prove Lemma \ref{LDTVector}.

We first present the condition of a general vector-valued function being a degree-$d$ polynomial.

\begin{lemma}
\label{lem:localtest}
    Let $|\mathbb{F}| > 2d$. The function $f:  \mathbb{F}^m\to \mathbb{F}^{\ell}$ is of degree-$d$ if for every $\vec{x}, \vec{h}\in \mathbb{F}^m$ there exists a degree-$d$ univariate polynomial $p_{\vec{x}, \vec{h}}: \mathbb{F}\to \mathbb{F}^{\ell}$ such that $p_{\vec{x}, \vec{h}}(i) = f(\vec{x} + i\vec{h})$ for every $i\in \mathbb{F}$.
\end{lemma}
\begin{proof}
    The cases when $\ell = 1$ is proved in \cite{RS96}. The general case follows directly from the fact that a vector-valued function $f = (f_1, f_2, \ldots, f_{\ell})$ is of degree-$d$ if and only if each $f_i$ is of degree-$d$.
\end{proof}

For $i\in \{0, 1, \ldots, d+1\}$, define $\alpha_i = (-1)^{i+1} \binom{d+1}{i}$.
Let $\mathbf{1}\in\mathbb{F}$ be the unity element. We view each number $i\in\mathbb{N}$ as an element $i\cdot\mathbf{1}\in\mathbb{F}$.
From the folklore relationship between a polynomial's degree and the number of times it needs to take difference on its values to make them all zero, we have the following lemma.

\begin{lemma}[\cite{RS96}, Folklore]
\label{lem:interpolation}
    A univariate polynomial $f: \mathbb{F}\to \mathbb{F}^{\ell}$ has degree $d < |\mathbb{F}|$ if and only if for every $e\in \mathbb{F}$ it holds that \[\sum_{i=0}^{d+1} \alpha_i f(e + i) = \vec{0}.\]
\end{lemma}

By combing Lemma \ref{lem:localtest} and \ref{lem:interpolation}, we have the following theorem.
\begin{theorem}
\label{thm:LDTcompleteness}
    Let $|\mathbb{F}| > 2d$. A function $f: \mathbb{F}^m\to \mathbb{F}^{\ell}$ is of degree-$d$ if and only if for every $\vec{x}, \vec{h}\in \mathbb{F}^m$ it holds that \[\sum_{i=0}^{d+1} \alpha_i f(\vec{x} + i\vec{h}) = \vec{0}.\]
\end{theorem}

We define our tester as follows. It selects $\vec{x}, \vec{h}\in \mathbb{F}^m$ uniformly at random, and query $f$ at points $\vec{x}, \vec{x} + \vec{h}, \vec{x} + 2\vec{h}, \ldots, \vec{x} + (d+1)\vec{h}$. It accepts if and only if \[\sum_{i=0}^{d+1} \alpha_i f(\vec{x} + i\vec{h}) = \vec{0}.\]

The proof of completeness follows directly from Theorem \ref{thm:LDTcompleteness}. For the soundness case, we have the following theorem.

\begin{theorem}
\label{thm:LDTsoundness}
    Let $\delta_0 = 1 / (d+2)^2$. If $f$ is not $\delta$-close to degree-$d$ polynomials, the above tester rejects with probability at least $\min(\delta, \delta_0) / 2$.
\end{theorem}
\begin{proof}

    It suffices to prove that: if the tester rejects with probability $\rho$ where $\rho < \delta_0 / 2$, then $f$ is $2\rho$-close to degree-$d$ polynomials.
    
    Now fix any function $f$ such that it passes the test with probability $1-
    \rho$, that is,\begin{equation}
    \label{equ:acceptprob}
        \Pr_{\vec{x}, \vec{h}\in \mathbb{F}^m}\left[\sum_{i=0}^{d+1} \alpha_i f(\vec{x} + i\vec{h}) = \vec{0}\right] \ge 1 - \rho
    \end{equation} where $\rho < \delta_0 / 2$. Define $\mathsf{MAJ}_{e\in S}\{v_e\}$ to be the most frequent value $v_e$ when $e\in S$ with ties broken arbitrarily, and \[g(\vec{x}) = \mathsf{MAJ}_{\vec{h}\in \mathbb{F}^m} \left\{ \sum_{i=1}^{d+1} \alpha_i f(\vec{x} + i\vec{h})\right\}= f(\vec{x}) + \mathsf{MAJ}_{\vec{h}\in \mathbb{F}^m} \left\{\sum_{i=0}^{d+1} \alpha_i f(\vec{x} + i\vec{h})\right\}.\]
    
    Consider the set of elements $\vec{x}$ such that \[\Pr_{\vec{h}\in \mathbb{F}^m} \left[\sum_{i=0}^{d+1} \alpha_i f(\vec{x} + i\vec{h}) = \vec{0}\right]\le 0.5.\] 
    By Inequality \ref{equ:acceptprob} and Markov's Inequality, the fraction of such elements is no more than $2\rho$. For remaining $\vec x$'s, we have \[\Pr_{\vec{h}\in \mathbb{F}^m} \left[\sum_{i=0}^{d+1} \alpha_i f(\vec{x} + i\vec{h}) = \vec{0}\right] > 0.5,\] which implies $g(\vec{x}) = f(\vec{x})$. Thus $f$ is $2\rho$-close to $g$ and it only remains to prove $g$ is of degree-$d$.
    
    We first claim that for every $\vec{x}$, $g(\vec{x})$ is the interpolation of $f(\vec{x} + i\vec{h})$ at $\vec{x}$ with high probability.
    Formally, we have the following claim.
    \begin{claim}
    \label{claim:probforx}
        For every $\vec{x}\in \mathbb{F}^m$, it holds that \[\Pr_{\vec{h}\in \mathbb{F}^m}\left[g(\vec{x}) = \sum_{j=1}^{d+1} \alpha_j f(\vec{x} + j\vec{h_2})\right]\ge 1 - 2(d+1)\rho.\]
    \end{claim}
    \begin{proof}

        Select $\vec{h}_1, \vec{h}_2\in \mathbb{F}^m$ independently and uniformly at random. Then for every $i, j\in [d+1]$, $\vec{x} + i\vec{h}_1$ and $\vec{h}_2$ (or $\vec{x} + j\vec{h}_2$ and $\vec{h}_1$) are also independently and uniformly distributed. It follows from Inequality \eqref{equ:acceptprob} that \[\Pr_{\vec{h}_1, \vec{h}_2\in \mathbb{F}^m} \left[f(\vec{x} + i\vec{h}_1) = \sum_{j=1}^{d+1} \alpha_j f((\vec{x} + i\vec{h}_1) + j\vec{h}_2)\right] \geq 1 - \rho\] and \[\Pr_{\vec{h}_1, \vec{h}_2\in \mathbb{F}^m} \left[f(\vec{x} + j\vec{h}_2) = \sum_{i=1}^{d+1} \alpha_i f((\vec{x} + j\vec{h}_2) + i\vec{h}_1)\right] \geq 1 - \rho.\] By a union bound over all $i,j\in [d+1]$, we have \[\Pr_{\vec{h}_1, \vec{h}_2\in \mathbb{F}^m} \left[\sum_{i=1}^{d+1} \alpha_i f(\vec{x} + i\vec{h}_1) = \sum_{i=1}^{d+1} \sum_{j=1}^{d+1} \alpha_i\alpha_j f(\vec{x} + i\vec{h}_1 + j\vec{h}_2)\right]\ge 1 - (d+1)\rho\] and \[\Pr_{\vec{h}_1, \vec{h}_2\in \mathbb{F}^m}\left[\sum_{j=1}^{d+1} \alpha_j f(\vec{x} + j\vec{h}_2) = \sum_{i=1}^{d+1} \sum_{j=1}^{d+1} \alpha_i\alpha_j f(\vec{x} + i\vec{h}_1 + j\vec{h}_2)\right]\ge 1 - (d+1)\rho.\] Taking union bound again we will get
        \begin{equation}
        \label{ineq:p2}
            \Pr_{\vec{h}_1, \vec{h}_2\in \mathbb{F}^m} \left[\sum_{i=1}^{d+1} \alpha_i f(\vec{x} + i\vec{h}_1) = \sum_{j=1}^{d+1} \alpha_j f(\vec{x} + j\vec{h}_2)\right]\ge 1 - 2(d+1)\rho.
        \end{equation}
        
        Fix $\vec x$, let $Z = \sum_{i=1}^{d+1} \alpha_i f(\vec{x} + i\vec{h})$ be a random variable where $\vec{h}$ is drawn from $\mathbb{F}^m$ uniformly. Suppose it takes value $Z_i$ with probability $p_i$ where $p_1\ge p_2\ge p_3\ge \ldots$. Note that $g(\vec{x})=Z_1$ and the claim is to give a lower bound on $p_1$ since $g(\vec{x})$ is the most frequent value of $Z$.
        
        The Inequality \eqref{ineq:p2} gives a lower bound on $\sum_{i} p_i^2$, and the claim follows from the fact that $\sum_i p_i^2\le \sum_i p_ip_1 = p_1$.
    \end{proof}
    
    Now fix $\vec{x}, \vec{h}\in \mathbb{F}^m$. Our goal is to prove $\sum_{i=0}^{d+1} \alpha_i g(\vec{x} + i\vec{h}) = \vec{0}.$ Select $\vec{h}_1, \vec{h}_2\in \mathbb{F}^m$ independently and uniformly at random. Then for every $i\in \{0, 1, \ldots, d+1\}$, $\vec{h}_1 + i\vec{h}_2$ is also uniformly distributed. According to Claim \ref{claim:probforx}, for every $i\in \{0, 1, \ldots, d+1\}$ we have \[\Pr_{\vec{h}_1, \vec{h}_2\in \mathbb{F}^m} \left[g(\vec{x} + i\vec{h}) = \sum_{j=1}^{d+1} \alpha_j f((\vec{x} + i\vec{h}) + j(\vec{h}_1 + i\vec{h}_2))\right]\ge 1-2(d+1)\rho.\]
    However, since for every $j\in [d+1]$, $\vec{x} + j\vec{h}_1$ and $\vec{h} + j\vec{h}_2$ are uniformly and independently distributed, according to Inequality \ref{equ:acceptprob}, we have \[\Pr_{\vec{h}_1, \vec{h}_2\in \mathbb{F}^m} \left[\sum_{i=0}^{d+1} \alpha_i f((\vec{x} + j\vec{h}_1) + i(\vec{h} + j\vec{h}_2)) = \vec{0}\right] \geq 1 - \rho\] By a union bound over $j\in [d+1]$, we get \[\Pr_{\vec{h}_1, \vec{h}_2\in \mathbb{F}^m} \left[\sum_{j=1}^{d+1} \alpha_j \sum_{i=0}^{d+1} \alpha_i f((\vec{x} + i\vec{h}) + j(\vec{h}_1 + i\vec{h}_2)) = \vec{0}\right]\ge 1 - (d+1)\rho\] and hence by a union bound over $i\in \{0, 1, \ldots, d+1\}$ we get
    \begin{align*}
        &\Pr_{\vec{h}_1, \vec{h}_2\in \mathbb{F}^m} \left[\sum_{i=0}^{d+1} \alpha_i g(\vec{x} + i\vec{h}) = \sum_{i=0}^{d+1} \alpha_i  \sum_{j=1}^{d+1} \alpha_j f((\vec{x} + i\vec{h}) + j(\vec{h}_1 + i\vec{h}_2)) = \vec{0}\right] \\
        &\ge 1 - (d+2)\cdot 2(d+1)\rho - (d+1)\rho \\
        & > 0
    \end{align*}
    where the last inequality comes from the fact $\rho\le 1/(2(d+2)^2)$. Since $\sum_{i=0}^{d+1} \alpha_i g(\vec{x} + i\vec{h})$ is independent of $\vec{h}_1$ and $\vec{h}_2$, it follows that \[\sum_{i=0}^{d+1} \alpha_i g(\vec{x} + i\vec{h}) = \vec{0}\] holds for all $\vec{x}, \vec{h}\in \mathbb{F}^m$. As a consequence, $f$ is $2\rho$-close to $g$, which is of degree-$d$, and our theorem is proved.
\end{proof}

Lemma \ref{LDTVector} directly follows from Theorem \ref{thm:LDTcompleteness} and \ref{thm:LDTsoundness}. 

\section{Detailed Proof of Lemma~\ref{lem:disperser}}
\label{appendix:disperser}

\begin{proof}[Proof of Lemma~\ref{lem:disperser}]

        Let $I_1 \ldots I_k$ denote independently and randomly chosen $\ell$-subsets of $[m]$. Consider any different indexes $1 \le i_1<\ldots<i_r \le k$, the probability that $|I_{i_1}\cup \ldots \cup I_{i_r}|<(1-\varepsilon) m$ can be bounded by the probability that there exists a subset $S$ of size $(1-\varepsilon)m$ such that $I_{i_1} \cup \ldots \cup I_{i_r} \subseteq S$. Take union bound over all $S$ and then over all possible $1 \le i_1<\ldots< i_r \le k$, we have
        \begin{align*}
    	&\ \Pr[\exists 1  \le i_1 < \ldots < i_r \le k, |I_{i_1} \cup \ldots \cup I_{i_r}| <(1-\varepsilon)m] \\ 
    	\le &\  k^r \cdot \Pr[\text{for fixed $1 \le i_1< \ldots < i_r \le k,|I_{i_1} \cup \ldots \cup I_{i_r}|<(1-\varepsilon)m$}] \\
    	\le &\ e^m \cdot \Pr[\text{for fixed $1 \le i_1< \ldots < i_r \le k, \exists S$, s.t. $|S| \le (1-\varepsilon) m$ and $I_{i_1} \cup \ldots \cup I_{i_r}\subseteq S$}] \\
    	\le &\ e^m \cdot 2^m \cdot \left(\frac{\binom{(1-\varepsilon)m}{\ell}}{\binom{m}{\ell}}\right)^r \\
    	\le &\ e^{2m} \cdot (1-\varepsilon)^{\ell r} \\
    	\le &\ e^{-m}.
        \end{align*}
\end{proof}

\end{document}